\documentclass[a4paper, 10pt, conference,onecolumn,draft]{ieeeconf}
\let\proof\relax   

\usepackage{amsthm,xpatch}
\usepackage{amsmath,amsfonts}
\usepackage{cite}
\usepackage{amssymb}
\usepackage{bm}
\usepackage{dsfont}
\usepackage{graphicx, subfigure}
\usepackage{color}
\usepackage{breqn}
\usepackage{mathtools}
\usepackage{bbm}
\usepackage{latexsym}
\usepackage[ruled, linesnumbered]{algorithm2e}
\usepackage{accents}
\usepackage{tikz}
\usepackage{multirow}
\usepackage{comment}

\usepackage[inline]{trackchanges}

\addeditor{as}

\newcommand\alex[1]{\add[as]{#1}}
\newcommand\alexn[1]{\notee[as]{#1}}

\newtheorem{lemma}{Lemma}
\newtheorem{theorem}{Theorem}

\newtheorem{example}{Example}

\newcommand*{\transpose}{%
  {\mathpalette\@transpose{}}%
}

\IEEEoverridecommandlockouts

\begin{document}

\newcommand{\SB}[3]{
\sum_{#2 \in #1}\biggl|\overline{X}_{#2}\biggr| #3
\biggl|\bigcap_{#2 \notin #1}\overline{X}_{#2}\biggr|
}

\newcommand{\Mod}[1]{\ (\textup{mod}\ #1)}

\newcommand{\overbar}[1]{\mkern 0mu\overline{\mkern-0mu#1\mkern-8.5mu}\mkern 6mu}

\makeatletter
\newcommand*\nss[3]{%
  \begingroup
  \setbox0\hbox{$\m@th\scriptstyle\cramped{#2}$}%
  \setbox2\hbox{$\m@th\scriptstyle#3$}%
  \dimen@=\fontdimen8\textfont3
  \multiply\dimen@ by 4             
  \advance \dimen@ by \ht0
  \advance \dimen@ by -\fontdimen17\textfont2
  \@tempdima=\fontdimen5\textfont2  
  \multiply\@tempdima by 4
  \divide  \@tempdima by 5          
  \ifdim\dimen@<\@tempdima
    \ht0=0pt                        
    \@tempdima=\fontdimen5\textfont2
    \divide\@tempdima by 4          
    \advance \dimen@ by -\@tempdima 
    \ifdim\dimen@>0pt
      \@tempdima=\dp2
      \advance\@tempdima by \dimen@
      \dp2=\@tempdima
    \fi
  \fi
  #1_{\box0}^{\box2}%
  \endgroup
  }
\makeatother

\makeatletter
\renewenvironment{proof}[1][\proofname]{\par
  \pushQED{\qed}%
  \normalfont \topsep6\p@\@plus6\p@\relax
  \trivlist
  \item[\hskip\labelsep
        \itshape
    #1\@addpunct{:}]\ignorespaces
}{%
  \popQED\endtrivlist\@endpefalse
}
\makeatother

\makeatletter
\newsavebox\myboxA
\newsavebox\myboxB
\newlength\mylenA

\newcommand*\xoverline[2][0.75]{%
    \sbox{\myboxA}{$\m@th#2$}%
    \setbox\myboxB\null
    \ht\myboxB=\ht\myboxA%
    \dp\myboxB=\dp\myboxA%
    \wd\myboxB=#1\wd\myboxA
    \sbox\myboxB{$\m@th\overline{\copy\myboxB}$}
    \setlength\mylenA{\the\wd\myboxA}
    \addtolength\mylenA{-\the\wd\myboxB}%
    \ifdim\wd\myboxB<\wd\myboxA%
       \rlap{\hskip 0.5\mylenA\usebox\myboxB}{\usebox\myboxA}%
    \else
        \hskip -0.5\mylenA\rlap{\usebox\myboxA}{\hskip 0.5\mylenA\usebox\myboxB}%
    \fi}
\makeatother

\xpatchcmd{\proof}{\hskip\labelsep}{\hskip3.75\labelsep}{}{}

\pagestyle{plain}

\title{\fontsize{20}{28}\selectfont Private Information Retrieval with Private Coded Side Information: The Multi-Server Case}

\author{Fatemeh Kazemi, Esmaeil Karimi, Anoosheh Heidarzadeh, and Alex Sprintson\thanks{The authors are with the Department of Electrical and Computer Engineering, Texas A\&M University, College Station, TX 77843 USA (E-mail: \{fatemeh.kazemi, esmaeil.karimi, anoosheh, spalex\}@tamu.edu).}}

\maketitle 

\thispagestyle{plain}

\begin{abstract}

In this paper, we consider the multi-server setting of Private Information Retrieval with Private Coded Side Information (PIR-PCSI) problem. In this problem, there is a database of $K$ messages whose copies are replicated across $N$ servers, and there is a user who knows a random linear combination of a random subset of $M$ messages in the database as side information. The user wishes to download one message from the servers, while protecting the identities of both the demand message and the messages forming the side information. We assume that the servers know the number of messages forming the user's side information in advance, whereas the indices of these messages and their coefficients in the side information are not known to any of the servers \emph{a priori}. 

Our goal is to characterize (or derive a lower bound on) the capacity, i.e., the maximum achievable download rate, for the \mbox{following} two settings. {In the first setting, the set of messages forming the linear combination available to the user as side information, does not include the user's demanded message. For this setting, we show that the capacity is equal to $\left(1+{1}/{N}+\dots+{1}/{N^{K-M-1}}\right)^{-1}$.} {In the second setting, the demand message contributes to the linear combination available to the user as side information, i.e., the demand message is one of the \mbox{messages} that form the user's side information. For this setting, we show that the capacity is lower-bounded by $\left(1+{1}/{N}+\dots+{1}/{N^{K-M}}\right)^{-1}$.} {The proposed achievability schemes and proof techniques leverage ideas from both our recent methods proposed for the single-server PIR-PCSI problem as well as the techniques proposed by Sun and Jafar for multi-server private computation problem.}


\end{abstract}

\section{Introduction}
In the Private Information Retrieval (PIR) problem, a database of $K$ messages are replicated at $N$ servers. There is a user who wishes to retrieve a single or multiple messages belonging to the database while protecting the identity of the demanded message(s) from any individual server~\cite{chor1995private,sun2017capacity,banawan2018capacity,sun2017}. In order to retrieve the desired message(s), the user generates one query for each server. Upon receiving the user's query, each server will return an answer to the user, which depends on the stored messages and the received query. To ensure that each server learns nothing about the identity of the message(s) being retrieved by the user, in an information theoretic sense, each query must be marginally independent of the desired message(s) index.

In a single-server setting or a multi-server setting when all servers can fully collude, the user must download the whole database to achieve privacy in the information-theoretic sense~\cite{chor1995private}. However, when the user has some side information about the messages in the database \cite{kadhe2017private,heidarzadeh,li2018single,chen2017capacity,li2018converse,shariatpanahi2018multi,heidarzadeh2018capacity,heidarzadeh2019capacity,kazemi2019multi,tandon2017capacity,wei2018cache,wei2018fundamental,heidarzadeh2019single,kazemi2019single} or when the servers do not fully collude~{\cite{sun2017capacity,banawan2018capacity,sun2017}}, the privacy can be achieved in a more efficient manner in terms of minimizing the download cost (i.e., the amount of information downloaded
from the server(s)).  

For the PIR problem in the presence of side information, two different types of privacy can be considered: (i)~\emph{$W$-privacy}, which requires that the identity of the user's demanded message(s) be protected, and (ii)~\emph{$(W,S)$-privacy}, which requires that  the identities of both the user's demanded message(s) and the message(s) in the user's side information be protected. When the side information is a random subset of messages, the problem is referred to as \emph{PIR with Side Information (PIR-SI)} or \emph{PIR with Private Side Information (PIR-PSI)} where $W$-privacy or $(W,S)$-privacy is required, respectively. The single-server settings of these problems were studied in~\cite{kadhe2017private,heidarzadeh,li2018single}, and their multi-server settings were studied in~\cite{li2018converse,chen2017capacity,shariatpanahi2018multi}. In~\cite{heidarzadeh2018capacity} and \cite{heidarzadeh2019capacity}, we studied the single-server setting of a related problem in which the side information is a random linear combination of a random subset of messages. This problem is referred to as \textit{PIR with Coded Side Information (PIR-CSI)} or \textit{PIR with Private Coded Side Information (PIR-PCSI)} when $W$-privacy or $(W,S)$-privacy is required, respectively. Also,  in \cite{kazemi2019multi}, we recently studied the multi-server setting of the PIR-CSI problem.


In this work, we consider the multi-server setting of the PIR-PCSI problem. In this setting, a database of $K$ messages is replicated across $N$ servers, and a user, who knows a random linear combination of a random subset of $M$ messages in the database, wishes to obtain a message by sending queries to the servers. The goal is to design a scheme that protects the identities of both the user's demanded message and the messages forming the user's side information, while minimizes the download cost. The servers are assumed to know the number of messages contributing to the user's side information beforehand. However, the indices and the coefficients of the messages in the user's side information are not known to the servers in advance. The motivation for this type of side information comes from several practical scenarios. For instance, the side information could have been obtained in advance from a trusted server with limited knowledge about the database, or through overhearing in a wireless network, or from the information locally stored in the user's cache. 

\subsection{Main Contributions}

We consider two settings of the PIR-PCSI problem depending on whether the user's demanded message is one of the messages forming the user's side information or not. We characterize (or derive a lower bound on) the capacity of each setting, where the capacity is defined as the supremum of all achievable rates (i.e., the inverse of the normalized download cost).  
In the first setting, the message demanded by the user is not one of the messages forming the user's side information. For this setting, we prove that the capacity is equal to $\left(1+{1}/{N}+\dots+{1}/{N^{K-M-1}}\right)^{-1}$. Interestingly, the capacity in this setting is equal to the capacity of multi-server PIR-PSI problem~\cite{chen2017capacity} in which $M$ uncoded messages are available at the user as side information. This result shows that there is no loss in capacity due to restricting the user's side information to one random linear combination of $M$ messages, instead of $M$ uncoded messages. 

The converse proof readily follows from the fact that the capacity of this setting is upper-bounded by the capacity of the multi-server PIR-PSI which is given by $\left(1+{1}/{N}+\dots+{1}/{N^{K-M-1}}\right)^{-1}$ (see~\cite[Theorem~1]{chen2017capacity}). 

For the achievability proof, we devise a new protocol that builds upon two existing achievability schemes for two different problems: (i) the Private Computation (PC) scheme of~\cite{sun2018capacity} for multi-server private computation where a user wishes to privately retrieve one arbitrary linear combination of the messages replicated at multiple servers, and (ii) our Specialized GRS Code scheme proposed in~\cite{heidarzadeh2019capacity} for single-server PIR-PCSI. 

The main ideas of our achievability scheme are as follows. First, the user utilizes the Specialized GRS Code scheme of~\cite{heidarzadeh2019capacity} for single-server PIR-PCSI to construct $K-M$ independent super-messages which are some linearly independent combinations of the original messages, to play the role of the original messages in a multi-server private computation problem. Then, the user and the $N$ servers leverage the PC scheme of~\cite{sun2018capacity} for the constructed $K-M$ super-messages in such a way that the user can privately download one of $\binom{K}{M+1}$ linear combinations of the $K-M$ super-messages where the support of each linear combination is a distinct subset of $[K]$ of size $M+1$.


Additionally, for the setting wherein the demanded message is one of the messages forming the user's side information, we show that the capacity is lower-bounded by ${\left(1+{1}/{N}+\dots+{1}/{N^{K-M}}\right)^{-1}}$. The proof is based on a new achievability scheme that leverages the PC scheme of~\cite{sun2018capacity} for multi-server private computation, combined with our Modified Specialized GRS Code scheme proposed in~\cite{heidarzadeh2019capacity} for single-server PIR-PCSI. 

\section{Problem Formulation}\label{sec:SN}
We denote random variables by bold letters and their realizations by non-bold letters. 
 For a positive integer $i$, let $[i]\triangleq \{1,\dots,i\}$. Let $\mathbb{F}_q$ be a finite field for some prime $q$, and let $\mathbb{F}_q^{\times} \triangleq \mathbb{F}_q\setminus \{0\}$ be the multiplicative group of $\mathbb{F}_q$. Let $\mathbb{F}_{q^m}$ be an extension field of $\mathbb{F}_q$ for some integer $m\geq 1$. 


Consider $N$ non-colluding identical servers, each of which stores $K$ messages $X_1,\dots,X_K$, where $\mathbf{X}_i$ is independently and uniformly distributed over $\mathbb{F}_{q^m}$, i.e., for all $i\in [K]$, it holds that \[{H(\mathbf{X}_i) = L\triangleq m\log_2 q}, \quad \text{and}\quad  {H(\mathbf{X}_1,\dots,\mathbf{X}_K) = KL}.\] 

Suppose that there is a user that wishes to retrieve a message $X_W$ from the servers for some $W\in [K]$, and has a linear combination ${Y^{[S,C]}\triangleq \sum_{i\in S} c_i X_i}$ for some $S \triangleq \{i_1,\dots,i_M\}\in \mathcal{S}$ and some ${C \triangleq \{c_{i_1},\dots,c_{i_M}\} \in \mathcal{C}}$, where $\mathcal{S}$ is the set of all $M$-subsets of $[K]$, and $\mathcal{C}$ is the set of all length-$M$ sequences with elements from $\mathbb{F}^{\times}_q$. We call $W$ the \emph{demand index}, $X_W$ the \emph{demand}, $Y^{[S,C]}$ the \emph{side information}, $S$ the \emph{side information index set}, and $M$ the \emph{side information size}. 

We assume that $\mathbf{S}$ is uniformly distributed over $\mathcal{S}$, and that $\mathbf{C}$ is uniformly distributed over $\mathcal{C}$. Also, two different models for the conditional distribution of $\mathbf{W}$ given $\mathbf{S}=S$ are considered: 
\begin{itemize}
    \item {Model I}: $\mathbf{W}$ is uniformly distributed over $[K]\setminus S$;
    \item {Model II}: $\mathbf{W}$ is uniformly distributed over $S$.
\end{itemize}

It is assumed that $1\leq M \leq K-1$ and $2\leq M\leq K$ for Model I and Model II, respectively. Note that for both models it holds that $\mathbf{W}$ is uniformly distributed over $[K]$. 
We assume that no server knows the realizations of $\mathbf{S},\mathbf{C},\mathbf{W}$ in advance. In contrast, we assume that all servers know the considered model (i.e., whether $\mathbf{W}\not\in \mathbf{S}$ or $\mathbf{W}\in \mathbf{S}$), the side information size $M$,  the distributions of $\mathbf{S}$ and $\mathbf{C}$, and the conditional distribution of $\mathbf{W}$ given $\mathbf{S}$. 



For any $S$, $C$, $W$, in order to retrieve $X_W$, the user generates $N$ queries $Q_1^{[W,S,C]},\dots, Q_N^{[W,S,C]}$, and sends to the $n$-th server the query $Q_n^{[W,S,C]}$. Each query $Q_n^{[W,S,C]}$ is assumed to be a (potentially stochastic) function of $W$, $S$, $C$, and $Y^{[S,C]}$. Upon receiving the query $Q_n^{[W,S,C]}$, the $n$-th server responds to the user with an answer $A_n^{[W,S,C]}$. The answer $A_n^{[W,S,C]}$ is a (deterministic) function of the query $Q_n^{[W,S,C]}$ and the messages in $X_{[K]}\triangleq \{X_1,\dots,X_K\}$. Note that for all $n\in [N]$, it holds that

\[(\mathbf{W},\mathbf{S}) \rightarrow (\mathbf{Q}_n^{[\mathbf{W},\mathbf{S},\mathbf{C}]},\mathbf{X}_{[K]}) \rightarrow \mathbf{A}_n^{[\mathbf{W},\mathbf{S},\mathbf{C}]}\] forms a Markov chain, and \[H(\mathbf{A}_n^{[\mathbf{W},\mathbf{S},\mathbf{C}]}| \mathbf{Q}_n^{[\mathbf{W},\mathbf{S},\mathbf{C}]},\mathbf{X}_{[K]}) = 0.\] 

The answers $A_1^{[W,S,C]},\dots,A_N^{[W,S,C]}$ from all servers along with the side information $Y^{[S,C]}$ and the queries $Q_1^{[W,S,C]},\dots,Q_N^{[W,S,C]}$ must enable the user to retrieve the demand $X_W$, i.e.,
\[H(\mathbf{X}_{\mathbf{W}}| \mathbf{A}^{[\mathbf{W},\mathbf{S},\mathbf{C}]},\mathbf{Q}^{[\mathbf{W},\mathbf{S},\mathbf{C}]}, \mathbf{Y}^{[\mathbf{S},\mathbf{C}]},\mathbf{W},\mathbf{S},\mathbf{C})=0,\]
where $A^{[W,S,C]}\triangleq \{A_1^{[W,S,C]},\dots,A_N^{[W,S,C]}\}$, and $Q^{[W,S,C]}\triangleq \{Q_1^{[W,S,C]},\dots,Q_N^{[W,S,C]}\}$. This condition is referred to as the \emph{recoverability condition}.

In addition, the queries $Q_1^{[W,S,C]}, \dots,Q_N^{[W,S,C]}$ must not reveal any information about the user's demand index $W$ and side information index set $S$ to any server,
\[I(\mathbf{W},\mathbf{S}; \mathbf{Q}_n^{[\mathbf{W},\mathbf{S},\mathbf{C}]},\mathbf{A}_n^{[\mathbf{W},\mathbf{S},\mathbf{C}]},\mathbf{X}_{[K]})=0\quad \forall n \in [N].\] This condition is referred to as the \emph{$(W,S)$-privacy condition.}
 
For both models (Model I and Model II), we would like to design a protocol for generating queries $\{Q_1^{[W,S,C]},\dots,Q_N^{[W,S,C]}\}$ for any given $W,S,C$. The protocol also prescribes, for all $n \in [N]$, how the $n$-th server generates the answer $A_n^{[W,S,C]}$, given $Q_n^{[W,S,C]}$ and $X_{[K]}$. 

A protocol that satisfies both the $(W,S)$-privacy and recoverability conditions for all $W,S,C$ with $W\not\in S$ (or $W\in S$), is referred to as a \emph{PIR-PCSI--I} (or \emph{PIR-PCSI--II}) \emph{protocol}. The problem of designing a PIR-PCSI--I (or PIR-PCSI--II) protocol is referred to as the \emph{PIR-PCSI--I} (or \emph{PIR-PCSI--II}) \emph{problem}. 



The \emph{rate} of a {PIR-PCSI--I} or {PIR-PCSI--II} protocol is defined as the ratio of the entropy of a message, i.e., $L$, to the total entropy of answers from all servers, i.e., $H(\mathbf{A}^{[\mathbf{W},\mathbf{S},\mathbf{C}]})$.  

The \emph{capacity} of the PIR-PCSI--I (PIR-PCSI--II) problem is defined as the supremum of rates over all PIR-PCSI--I (PIR-PCSI--II) protocols. We denote by ${C_{(W,S)-\text{\it I}}}$ the capacity of the PIR-PCSI--I problem, and denote by ${C_{(W,S)-\text{\it II}}}$ the capacity of the PIR-PCSI--II problem.  

In this work, our goal is to characterize (or derive lower bounds on) ${C_{(W,S)-\text{\it I}}}$ and ${C_{(W,S)-\text{\it II}}}$, and to design {PIR-PCSI--I} and {PIR-PCSI--II} protocols that achieve the capacity (or the derived lower bound on the capacity). 

\section{Main Results}

In this section, we present our main results. Theorem~\ref{thm:PIRPCSI-I} characterizes the capacity of the PIR-PCSI--I problem ${C_{(W,S)-\text{\it I}}}$, and Theorem~\ref{thm:PIRPCSI-II} presents a lower-bound on the capacity of the PIR-PCSI--II problem ${C_{(W,S)-\text{\it II}}}$. The proofs of theorems~\ref{thm:PIRPCSI-I} and~\ref{thm:PIRPCSI-II} are given in sections~\ref{sec:PIR-PCSI-I} and~\ref{sec:PIR-PCSI-II}, respectively.

\begin{theorem}\label{thm:PIRPCSI-I}
The capacity of the PIR-PCSI--I problem with $N$ servers, $K$ messages, and side information size ${1\leq M \leq K-1}$ is given by
\[
{C_{(W,S)-\text{\it I}}}= \left(1+\frac{1}{N}+\dots+\frac{1}{N^{K-M-1}}\right)^{-1}.\]

\end{theorem}

Interestingly, this result indicates that the capacity of multi-server PIR-PCSI--I, i.e., ${C_{(W,S)-\text{\it I}}}$, is equal to the capacity of the multi-server PIR-PSI~\cite{chen2017capacity} where $M$ uncoded messages are available at the user as side information. Note that having only a random linear combination of $M$ messages as side information instead of $M$ uncoded messages, cannot increase the capacity which implies the converse. Thus, to complete the proof of Theorem~\ref{thm:PIRPCSI-I}, we only need to prove the achievability which is presented in Section~\ref{sec:PIR-PCSI-I}.
Notably, our results show that having only one random linear combination of messages instead of multiple uncoded messages does not decrease the capacity, either. 

\begin{theorem}\label{thm:PIRPCSI-II}
The capacity of the PIR-PCSI--II problem with $N$ servers, $K$ messages, and side information size ${2 \leq M\leq K}$ is lower-bounded by
\[
{C_{(W,S)-\text{\it II}}} \geq \left(1+\frac{1}{N}+\dots+\frac{1}{N^{K-M}}\right)^{-1}.
\]

\end{theorem}

This result is interesting because it shows that the lower-bound on the capacity of the multi-server PIR-PCSI--II is the same as the capacity of multi-server PIR-SI when the size of side information is $M-1$. That is, 
having a side information which is only a random linear combination of $M$ messages including the demanded message would be at least as effective as knowing $M-1$ messages separately in terms of minimizing the download cost. 
For the proof, we construct a {PIR-PCSI--II} protocol that achieves the capacity lower-bound of Theorem~\ref{thm:PIRPCSI-II}. It should be noted that the tightness of this lower bound remains open in general. 



\section{The ~PIR-PCSI-I ~Problem}\label{sec:PIR-PCSI-I}

In this section, we complete the proof of Theorem~\ref{thm:PIRPCSI-I} by proposing an achievability scheme for arbitrary $N$, ${K\geq 1}$ and ${0\leq M\leq K-1}$ that achieves the rate $\left(1+{1}/{N}+\dots+{1}/{N^{K-M-1}}\right)^{-1}$. The proposed protocol, referred to as the \emph{Multi-Server PIR-PCSI--I protocol}, is a non-trivial combination of the Specialized GRS Code scheme of~\cite{heidarzadeh2019capacity} for single-server PIR-PCSI problem and the Private Computation (PC) scheme of~\cite{sun2018capacity} for multi-server private computation problem. 

For the proposed protocol, we assume that $q\geq K$, and each message $X_i$ consists of $m= N^{\binom{K}{M+1}}$ symbols over $\mathbb{F}_q$.\vspace{0.2cm}


\textbf{Multi-Server PIR-PCSI--I protocol:} The protocol consists of the following five steps:

\textit{\textbf{Step 1:}} The user utilizes the Specialized GRS Code protocol proposed in \cite{heidarzadeh2019capacity} to first construct a polynomial ${p(x) = \sum_{i=0}^{K-M-1} p_i x^i \triangleq \prod_{i\not\in S\cup W} (x-\omega_i)}$ where $\omega_1,\dots,\omega_K$ are $K$ arbitrarily chosen distinct elements from $\mathbb{F}_q$, and then construct $r \triangleq K-M$ vectors $\underline{u}_1,\dots,\underline{u}_{r}$, each of length $K$, such that ${\underline{u}_{i} = [\beta_{1}\omega_1^{i-1} ,\dots,\beta_{K}\omega_K^{i-1}]}$ for $i\in [r]$, where $\beta_j=\frac{c_j}{p(\omega_j)}$ for $j\in S$, and $\beta_j$ is a randomly chosen element from $\mathbb{F}_q^{\times}$ for $j\not\in S$. 



\textit{\textbf{Step 2:}} Let $\hat{X}_i\triangleq \sum_{j=1}^{K} \beta_j\omega_j^{i-1} X_{j}$ for $i\in [r]$. Each $\hat{X}_i$ is referred to as a \emph{super-message}. Note that the vector $\underline{u}_i$ (constructed in Step~1) is the vector of coefficients of the messages $\{X_i\}_{i\in [K]}$ in the super-message $\hat{X}_i$. Let $F\triangleq\binom{K}{M+1}$, and let $J_1,J_2,\dots,J_F$ be the collection of all $(M+1)$-subsets of $[K]$ in a lexicographical order. The structure of the Specialized GRS Code protocol \cite{heidarzadeh2019capacity} ensures that for each $J_f$, $f \in [F]$, there exist exactly $q-1$ linear combinations $Z^1_f,Z^2_f,\dots,Z^{q-1}_f$ of the messages $\{X_i\}_{i \in J_f}$ with (non-zero) coefficients from $\mathbb{F}^{\times}_q$, such that for every $k\in [q-1]$, $Z^k_f$ can be written as a linear combination of the super-messages $\hat{X}_1,\dots,\hat{X}_r$. Let $\underline{v}^{k}_f\triangleq [v^{k}_{f,1},\dots,v^{k}_{f,r}]$ be a vector of length $r$ such that $Z^{k}_f = \sum_{i=1}^{r} v^{k}_{f,i} \hat{X}_i$. Note that, for each $f \in [F]$, $Z^1_f,Z^2_f,\dots,Z^{q-1}_f$ are the same up to a scalar multiple, i.e., for each $k\in [q-1]$, $Z^{k}_f = \alpha_k Z^{1}_f$, or equivalently, $\underline{v}^{k}_f = \alpha_k \underline{v}^{1}_f$, for some distinct $\alpha_k\in \mathbb{F}^{\times}_q$. 
For each $f\in [F]$, let $i_f\triangleq \min(J_f)$. Note also that for every $f\in [F]$, there exists a unique $k_f\in [q-1]$ such that the coefficient of the message $X_{i_f}$ in the linear combination $Z^{k_f}_{f}$ is equal to $1$. The user then constructs $F$ vectors $\underline{v}_1,\dots,\underline{v}_F$, each of length $r$, such that $\underline{v}_f = \underline{v}^{k_f}_f$. (Note that the above procedure dictates a specific choice of the coefficient vectors $\underline{v}_f$. However, for each $f\in [F]$, the vector $\underline{v}_f$ can be chosen arbitrarily from the set of vectors $\{\underline{v}^{k}_f\}_{k\in [q-1]}$.) Let $Z_f\triangleq Z^{k_f}_f$ for $f\in [F]$. Each $Z_f$ is referred to as a (linear) \emph{function}. Note that $\underline{v}_f$ is the vector of coefficients of the super-messages $\{\hat{X}_i\}_{i\in [r]}$ in the function $Z_f$. 

\textit{\textbf{Step 3:}} The user sends to all servers the vectors $\underline{u}_1,\dots,\underline{u}_{r}$ (associated with the super-messages $\hat{X}_1,\dots,\hat{X}_r$), and the vectors $\underline{v}_1,\dots,\underline{v}_{F}$ (associated with the functions $Z_1,\dots,Z_F$). It is noteworthy that the user needs only to send the vectors $\{\underline{u}_i\}_{i\in [r]}$ to all servers, and each server can construct the vectors $\{\underline{v}_f\}_{f\in [F]}$ by using $\{\underline{u}_i\}_{i\in [r]}$ (according to the procedure described in Step~2).

\textit{\textbf{Step 4:}} The user and the servers leverage the PC scheme of~\cite{sun2018capacity} with $r$ (independent) messages and $F$ (linear) functions of these messages in order for the user to privately retrieve one of these functions. In particular, the $r=K-M$ super-messages $\{\hat{X}_i\}_{i\in [r]}$ and the $F$ functions $\{Z_f\}_{f\in [F]}$ play the role of the original messages and the functions in the PC scheme, respectively, and the user is interested in retrieving the function $Z_{f^{*}}$ privately, where $Z_{f^{*}}$ is an $\mathbb{F}^{\times}_q$-linear combination (i.e., a linear combination with non-zero coefficients only) of the messages $\{X_i\}_{i \in W\cup S}$. (By the construction, there exists one (and only one) function $Z_f$ among $Z_1,\dots,Z_F$ such that $Z_f$ is an $\mathbb{F}_q^{\times}$-linear combination of the messages $\{X_i\}_{i \in W\cup S}$.) To be more specific, each server first constructs the super-messages $\{\hat{X}_i\}_{i\in [r]}$ by using the coefficient vectors $\{\underline{u}_i\}_{i\in [r]}$ (defined in Step~3), and then constructs the functions $\{Z_f\}_{f\in [F]}$ by using the super-messages $\{\hat{X}_i\}_{i\in [r]}$ and the coefficient vectors $\{\underline{v}_f\}_{f\in [F]}$ (defined in Step~3). Note that each function $Z_f$ for $f\in [F]$ consists of $m=N^F$ $\mathbb{F}_q$-symbols where $N$ is the number of servers. Then, each server sends to the user $m({{1}/{N}+{1}/{N^2}+\dots+{1}/{N^{K-M}}})$ carefully designed linear combinations of all $\mathbb{F}_q$-symbols associated with all functions $\{Z_f\}_{f\in [F]}$. The details of the design of the user's query to each server as well as the linear combinations transmitted by each server (which also depend on the query of the user) can be found in~\cite[Section~4]{sun2018capacity}.

\textbf{\textbf{Example 1.}} (Multi-Server PIR-PCSI--I protocol) Assume that there are $N = 2$ servers, $K=4$ messages from $\mathbb{F}_{5^{16}}$, and $M=2$. Note that each message consists of ${m= N^{\binom{K}{M+1}}= 16}$ symbols from $\mathbb{F}_{5}$. Suppose that the user demands the message $X_1$ and has a coded side information $Y={X_2+X_3}$, i.e., $W = 1$, $S= \{2,3\}$, and ${C = \{1,1\}}$ (i.e., ${c_{2}=1, c_{3}=1}$). 

First, the user picks ${K=4}$ distinct elements $\omega_1,\dots,\omega_4$ from $\mathbb{F}_5$. Suppose that the user chooses $\omega_1=0$, $\omega_2=1$, $\omega_3=2$, $\omega_4=3$. Then, the user constructs the polynomial ${{p(x) = \prod_{i\not\in S\cup W} (x-\omega_i)}=x-\omega_4=x-3}$. The user then computes $\beta_j$ for $j\in S$, i.e., $\beta_2$ and $\beta_3$, by setting ${\beta_2=\frac{c_2}{p(\omega_2)}=2}$ and ${\beta_3=\frac{c_3}{p(\omega_3)}=4}$, and chooses $\beta_j$ for $j\not\in S$, i.e., $\beta_1$ and $\beta_4$, at random (from $\mathbb{F}^{\times}_5$). Assume that the user chooses $\beta_1=1$ and $\beta_4=2$. Then, the user constructs ${r=K-M=2}$ vectors $\underline{u}_1$ and $\underline{u}_2$, each of length ${K=4}$, such that ${\underline{u}_i=[\beta_1\omega_1^{i-1},\dots,\beta_K\omega_K^{i-1}]}$ for $i\in \{1,2\}$. That is, the user constructs ${\underline{u}_1=[1,2,4,2]}$ and ${\underline{u}_2=[0,2,3,1]}$. For set $J_1=\{1,2,3\}$, there exist exactly ${q-1=4}$ vectors ${\underline{v}^{k}_1 = [k,3k]}$ for ${k\in \{1,\dots,4\}}$ such that ${k\underline{u}_1+3k\underline{u}_2=k[1,3,3,0]}$.

It should be noted that there exists no other vector ${\underline{v}=[v_1,v_2]}$ such that the support of the vector ${v_1\underline{u}_1 + v_2\underline{u}_2}$ is ${J_1 = \{1,2,3\}}$. Note that the coefficient of the message ${X_{i_1}=X_1}$ (i.e., ${i_1=\min(J_1)=1}$) in the \mbox{function} ${Z_1}$ is equal to $1$ when ${k=1}$. Thus, the user constructs the vector ${\underline{v}_1=\underline{v}^{1}_1=[1,3]}$. Similarly, the user \mbox{constructs} the vectors ${\underline{v}_2=[1,2]}$, ${\underline{v}_3=[1,4]}$ and ${\underline{v}_4=[0,3]}$. Then, the user sends to all servers the vectors $\underline{u}_1$ and $\underline{u}_{2}$ (associated with the super-messages $\hat{X}_1$ and $\hat{X}_2$), and the vectors ${\underline{v}_1,\dots,\underline{v}_{4}}$ (associated with the functions ${Z_1,\dots,Z_4}$). Using the coefficient vectors $\underline{u}_1$ and $\underline{u}_{2}$, each server first constructs the two super-messages ${\hat{X}_1=X_1+2X_2+4X_3+2X_4}$ and ${\hat{X}_2=2X_2+3X_3+X_4}$. Then, it constructs the functions ${Z_1,\dots,Z_4}$ using the super-messages $\hat{X}_1$ and $\hat{X}_2$ and the coefficient vectors ${\underline{v}_1,\dots,\underline{v}_{4}}$ as follows: 
\[
\begin{array}{lcl}
Z_1=\hat{X}_1+3\hat{X}_2=X_1+3X_2+3X_3\\
Z_2=\hat{X}_1+2\hat{X}_2=X_1+X_2+4X_4\\
Z_3=\hat{X}_1+4\hat{X}_2=X_1+X_3+X_4\\
Z_4=3\hat{X}_2=X_2+4X_3+3X_4\\   
\end{array}
\]

Finally, the user and the servers apply the PC scheme of~\cite{sun2018capacity} for two super-messages $\hat{X}_1$, $\hat{X}_2$ in order for the user to privately retrieve the function ${Z_1}$. (Note that among the functions $Z_1,\dots,Z_4$, only $Z_1$ is an $\mathbb{F}^{\times}_5$-linear combination of the messages $\{X_i\}_{i\in W\cup S} = \{X_1,X_2,X_3\}$.) The details of the PC scheme for this example are as follows. Let $\pi: [16]\rightarrow [16]$ be a randomly chosen permutation. Let $u_f(i)\triangleq\sigma_i Z_f(\pi(i))$ for $f \in [4]$ and ${i \in [16]}$, where $Z_f(\pi(i))$ is the $\pi(i)$-th $\mathbb{F}_5$-symbol of $Z_f$, and $\sigma_i$ is a randomly chosen element from $\{-1,+1\}$. 
For simplifying the notation, let $(a_i,b_i,c_i,d_i)=(u_1(i),u_2(i),u_3(i),u_4(i))$ for all ${i \in [16]}$. The user then queries $15$ carefully designed linear combinations of the symbols $\{\{a_i\}_{i\in [16]},\{b_i\}_{i\in [16]},\{c_i\}_{i\in [16]},\{d_i\}_{i\in [16]}\}$, as given in Table~\ref{table1} \cite{sun2018capacity}, from each of the servers (S1 and S2). 

As shown in~\cite{sun2018capacity}, among the $15$ symbols queried from S1 (or S2), based on the information obtained from S2 (or S1), $3$ symbols are redundant. For instance, consider the $15$ symbols queried from S1. (Similar observations can be made regarding the queries from S2.) Among the $4$ symbols $\{a_1,b_1,c_1,d_1\}$, any $2$ symbols suffice to recover the other $2$ symbols. For example, $c_1$ and $d_1$ can be obtained from $a_1$ and $b_1$. (Note that $Z_3$ and $Z_4$ can be written as a linear combination of $Z_1$ and $Z_2$.) Thus, the server S1 needs to send two arbitrary symbols from $\{a_1,b_1,c_1,d_1\}$. In addition, given any $2$ symbols from $\{a_2,b_2,c_2,d_2\}$, any $5$ symbols among the $6$ symbols $\{{a_3-b_2},{a_4-c_2},{a_5-d_2},{b_4-c_3},{b_5-d_3},{c_5-d_4}\}$ queried from S1 would suffice to recover the remaining symbol. For example, ${c_5-d_4}$ can be obtained from the symbols $\{a_3-b_2,{a_4-c_2},{a_5-d_2},{b_4-c_3},{b_5-d_3},b_2,d_2\}$ (for details, see~\cite[Section~5.1]{sun2018capacity}). 
Thus, each of the servers S1 and S2 needs to send to the user only $12$ symbols. In particular, S1 transmits $2$ arbitrary symbols from $\{a_1,b_1,c_1,d_1\}$, $5$ arbitrary symbols from $\{{a_3-b_2},{a_4-c_2},{a_5-d_2},{b_4-c_3},{b_5-d_3},{c_5-d_4}\}$, and the $4$ symbols $\{a_9-b_7+c_6,a_{10}-b_8+d_6,a_{11}-c_8+d_7,b_{11}-c_{10}+d_9\}$, and the symbol $\{a_{15}-b_{14}+c_{13}-d_{12}\}$; and S2 transmits $2$ arbitrary symbols from $\{a_2,b_2,c_2,d_2\}$, $5$ arbitrary symbols from $\{a_6-b_1,{a_7-c_1},{a_8-d_1},{b_7-c_6},{b_8-d_6},{c_8-d_7}\}$, and the $4$ symbols $\{a_{12}-b_4+c_3,a_{13}-b_5+d_3,a_{14}-c_5+d_4,b_{14}-c_{13}+d_{12}\}$, and the symbol $\{a_{16}-b_{11}+c_{10}-d_{9}\}$. 

\begin{table}[t!]
    \caption{The queries of the PC protocol for $N=2$, $2$ super-messages, $F=4$, when the user demands ${Z_1}$\cite{sun2018capacity}.}
    \label{table1}
    \centering
    \scalebox{1.25}{
\begin{tabular}{ |c|c| } 
 \hline 
 S1 & S2 \\ 
 \hline
 $a_{1},b_{1},c_{1},d_{1}$ & $a_{2}, b_{2},c_{2},d_{2}$ \\ 
 \hline
 $a_{3}-b_{2}$ & $a_{6}-b_{1}$\\ 
 $a_{4}-c_{2}$ & $a_{7}-c_{1}$\\ 
 $a_{5}-d_{2}$ & $a_{8}-d_{1}$ \\ 
  $b_{4}-c_{3}$ & $b_{7}-c_{6}$\\ 
  $b_{5}-d_{3}$ & $b_{8}-d_{6}$\\ 
  $c_{5}-d_{4}$ & $c_{8}-d_{7}$\\ 
 \hline
  $a_{9}-b_{7}+c_{6}$ & $a_{12}-b_{4}+c_{3}$\\ 
  $a_{10}-b_{8}+d_{6}$ & $a_{13}-b_{5}+d_{3}$\\ 
  $a_{11}-c_{8}+d_{7}$ & $a_{14}-c_{5}+d_{4}$\\ 
  $b_{11}-c_{10}+d_{9}$ & $b_{14}-c_{13}+d_{12}$\\ 
  \hline
  $a_{15}-b_{14}+c_{13}-d_{12}$ & $a_{16}-b_{11}+c_{10}-d_{9}$\\ 
  \hline
\end{tabular}}
\end{table}

From the answers by the servers, the user obtains all $16$ symbols $a_1,\dots,a_{16}$, and accordingly, all $16$ symbols of $Z_1$. (Note that $a_i = u_1(i)=\sigma_i Z_1(\pi(i))$ for $i\in [16]$.) From $Z_1$ ($=X_1+3X_2+3X_3$), the user can decode the desired message $X_1$ by subtracting off the contribution of their side information $X_2+X_3$.

In order to retrieve $X_1$ which consists of $16$ symbols (over $\mathbb{F}_5$), according to the proposed protocol, the user downloads $24$ symbols (over $\mathbb{F}_5$) from both servers, and hence the rate of the proposed protocol is $16/24=2/3$.

Note that for every $3$-subset $\{X_{j_1},X_{j_2},X_{j_3}\}$ of the messages $\{X_i\}_{i\in [4]}$, in the proposed protocol there exists one (and only one) linear combination $Z_f$ for some $f \in [4]$ of the messages $X_{j_1},X_{j_2},X_{j_3}$. On the other hand, the PC scheme guarantees that no server can obtain any information about the index ($f$) of the linear combination $Z_f$ being requested by the user. Thus, the proposed scheme satisfies the $(W,S)$-privacy condition, as desired.  


\begin{lemma}\label{lemma1}
The Multi-Server PIR-PCSI--I protocol satisfies the recoverability and (W,S)-privacy conditions, and achieves the rate ${C_{(W,S)-\text{\it I}}}= \left(1+{1}/{N}+\dots+{1}/{N^{K-M-1}}\right)^{-1}$.
\end{lemma}

\begin{proof}

Since the messages $\mathbf{X}_{[K]}$ are uniformly and independently distributed over $\mathbb{F}_{q^m}$, and
$\{\hat{X}_1,\dots,\hat{X}_r\}$ are linearly independent combinations of the messages in $X_{[K]}$, thus $\{\hat{\mathbf{X}}_1,\dots,\hat{\mathbf{X}}_r\}$ are uniformly and independently distributed over $\mathbb{F}_{q^m}$ as well, i.e., ${H(\hat{\mathbf{X}}_1) = \dots = H(\hat{\mathbf{X}}_r) = m\log q=L}$. Hence, the rate of the Multi-Server PIR-PCSI--I protocol is the same as the rate of the PC protocol for $N$ servers and ${K-M}$ messages, which is given by $\left(1+{1}/{N}+\dots+{1}/{N^{K-M-1}}\right)^{-1}$ (see \cite[Theorem~1]{sun2018capacity}).

From the step $4$ of the Multi-Server PIR-PCSI--I protocol, it is evident that the recoverability condition is satisfied. 
The proof of the $(W,S)$-privacy of the proposed protocol is as follows. 
The PC protocol protects the privacy of the function (linear combination) requested by the user. That is, given the query, no server can obtain any information about the index of the function requested by the user. Consider an arbitrary server $n\in [N]$, and an arbitrary query $Q_n$ to server $n$, generated by the proposed protocol. Thus, given ${\mathbf{Q}^{[\mathbf{W},\mathbf{S},\mathbf{C}]}_n = Q_n}$, from the perspective of server $n$, every function ${Z_f}$ for ${f \in [F]}$ is equally likely to include the demanded message. We denote the support of $Z_f$ by $\mathcal{Z}_f$, i.e., $\mathcal{Z}_f$ is the set of all indices $i\in [K]$ such that $X_i$ has a non-zero coefficient in the linear combination $Z_f$. Thus, for all $f\in [F]$, we have 
\begin{equation}\label{eq:P1}
\Pr(\mathbf{W}\in \mathcal{Z}_f|\mathbf{Q}^{[\mathbf{W},\mathbf{S},\mathbf{C}]}_n = Q_n) = \frac{1}{\binom{K}{M+1}},
\end{equation} noting that $F=\binom{K}{M+1}$. 
Note that any given index $W'\in [K]$ is in the support of exactly $\binom{K-1}{M}$ functions $Z_f$, $f\in [F]$. 
For any given $f\in [F]$, given $\mathbf{Q}^{[\mathbf{W},\mathbf{S},\mathbf{C}]}_n = Q_n$ and $\mathbf{W}\in \mathcal{Z}_f$, from the perspective of server $n$, every index $W'\in \mathcal{Z}_f$ is equally likely to be the demand index. That is, for all $f\in [F]$, we have 
\begin{equation}\label{eq:P2}
\Pr(\mathbf{W}=W'|\mathbf{Q}^{[\mathbf{W},\mathbf{S},\mathbf{C}]}_n = Q_n,\mathbf{W}\in \mathcal{Z}_f) = \begin{cases} \frac{1}{M+1}, & W'\in \mathcal{Z}_f,\\ 0, & \text{otherwise}.\end{cases}
\end{equation} Furthermore, for any given $f\in [F]$ and $W'\in \mathcal{Z}_f$, we have 
\begin{equation}\label{eq:P3}
\Pr(\mathbf{S}=S'|\mathbf{Q}^{[\mathbf{W},\mathbf{S},\mathbf{C}]}_n = Q_n,\mathbf{W}\in \mathcal{Z}_f,\mathbf{W} = W') = \begin{cases} 1, & S'= \mathcal{Z}_f\setminus \{W'\},\\ 0, & \text{otherwise}.\end{cases}
\end{equation} Consider arbitrary $W'\in [K]$ and $S'\subset [K]\setminus \{W'\}, |S'| = M$. Let $f'\in [F]$ be the (unique) index such that $\mathcal{Z}_{f'} = W'\cup S'$. It is easy to see that $\Pr(\mathbf{W}=W',\mathbf{S}=S',\mathbf{W}\in \mathcal{Z}_f|\mathbf{Q}^{[\mathbf{W},\mathbf{S},\mathbf{C}]}_n = Q_n) = 0$ for all $f\in [F], f\neq f'$. Thus, by using~\eqref{eq:P1}-\eqref{eq:P3}, we can write 
\begin{align}
&\Pr(\mathbf{W}=W',\mathbf{S}=S'|\mathbf{Q}^{[\mathbf{W},\mathbf{S},\mathbf{C}]}_n = Q_n)\nonumber \\  
& = \sum_{f\in [F]} \Pr(\mathbf{W}=W',\mathbf{S}=S',\mathbf{W}\in \mathcal{Z}_f|\mathbf{Q}^{[\mathbf{W},\mathbf{S},\mathbf{C}]}_n = Q_n)\nonumber \\
& = \Pr(\mathbf{W}=W',\mathbf{S}=S',\mathbf{W}\in \mathcal{Z}_{f'}|\mathbf{Q}^{[\mathbf{W},\mathbf{S},\mathbf{C}]}_n = Q_n)\nonumber \\
& = \Pr(\mathbf{W}\in \mathcal{Z}_{f'}|\mathbf{Q}^{[\mathbf{W},\mathbf{S},\mathbf{C}]}_n = Q_n)\times \Pr(\mathbf{W} = W'|\mathbf{Q}^{[\mathbf{W},\mathbf{S},\mathbf{C}]}_n = Q_n, \mathbf{W}\in \mathcal{Z}_{f'})\nonumber\\ & \quad \times \Pr(\mathbf{S} = S'|\mathbf{Q}^{[\mathbf{W},\mathbf{S},\mathbf{C}]}_n = Q_n, \mathbf{W}\in \mathcal{Z}_{f'},\mathbf{W} = W')\nonumber\\
& = \frac{1}{\binom{K}{M+1}}\times \frac{1}{M+1}\times 1 \nonumber \\ 
& = \frac{M! (K-M-1)!}{K!} \label{eq:P4}
\end{align} On the other hand, we have 
\begin{align}
&\Pr(\mathbf{W}= W',\mathbf{S}= S')
\nonumber \\ &= \Pr(\mathbf{W}= W') \times \Pr(\mathbf{S}= S'| \mathbf{W}= W') \nonumber \\ & =\frac{1}{K} \times \frac{1}{\binom{K-1}{M}} \nonumber\\
& = \frac{M! (K-M-1)!}{K!}. \label{eq:P5}
\end{align} From~\eqref{eq:P4} and~\eqref{eq:P5}, for any $W'\in [K]$ and $S'\subset [K]\setminus \{W'\}, |S'|=M$, we have \[\Pr(\mathbf{W}= W',\mathbf{S}= S'| \mathbf{Q}^{[\mathbf{W},\mathbf{S},\mathbf{C}]}_n=Q_n) = \Pr(\mathbf{W}= W',\mathbf{S}= S').\] This completes the proof of $(W,S)$-privacy of the proposed protocol.
\end{proof}

\section{The ~PIR-PCSI-II ~Problem}\label{sec:PIR-PCSI-II}

In this section, we prove the result of Theorem~\ref{thm:PIRPCSI-II} 
by constructing a PIR-PCSI--II protocol, referred to as the \emph{Multi-Server PIR-PCSI--II protocol}, for arbitrary $N$, ${K\geq 2}$ and ${2 \leq M\leq K}$ that achieves the rate ${\left(1+{1}/{N}+\dots+{1}/{N^{K-M}}\right)^{-1}}$. 

For the proposed protocol, we assume that $q\geq K$, and each message is comprised of ${m = N^{\binom{K}{M}}}$ symbols over $\mathbb{F}_q$.\vspace{0.2cm} 


\textbf{Multi-Server PIR-CSI--II protocol:} The protocol consists of four steps, where the steps 2-4 are the same as the steps 2-4 in the Multi-Server PIR-PCSI--I protocol, except that $M$ is replaced with $M-1$ everywhere. The step~1 of the proposed protocol is as follows: 

\textit{\textbf{Step 1:}} The user utilizes the Modified Specialized GRS Code protocol proposed in \cite{heidarzadeh2019capacity} to first construct a polynomial  ${p(x) = \sum_{i=0}^{K-M} p_i x^i \triangleq \prod_{i\not\in S} (x-\omega_i)}$ where $\omega_1,\dots,\omega_K$ are $K$ arbitrarily chosen distinct elements from $\mathbb{F}_q$, and then construct $r \triangleq K-M+1$ vectors ${\underline{u}_1,\dots,\underline{u}_{r}}$, each of length $K$, such that ${\underline{u}_{i} = [\beta_1\omega_1^{i-1} ,\dots,\beta_K\omega_K^{i-1}]}$ for $i\in [r]$, where $\beta_j=\frac{c_j}{p(\omega_j)}$ for $j\in S\setminus W$, $\beta_W=\frac{c}{p(\omega_W)}$ where $c$ is chosen uniformly at random from $\mathbb{F}^{\times}_q\setminus \{c_W\}$, and $\beta_j$ is a randomly chosen element from $\mathbb{F}_q^{\times}$ for $j\not\in S$. 

\begin{lemma}
The Multi-Server PIR-PCSI--II protocol satisfies the recoverability and (W,S)-privacy conditions, and achieves the rate $\left(1+{1}/{N}+\dots+{1}/{N^{K-M}}\right)^{-1}$.
\end{lemma}

\begin{proof}
The proof is similar to the proof of Lemma \ref{lemma1}, and hence omitted to avoid repetition.
\end{proof}

\section{Conclusion}
In this paper, we studied the multi-server setting of the Private Information Retrieval with Private Coded Side Information (PIR-PCSI) problem. In this problem, there is a database of $K$ messages replicated across $N$ servers, and there is a user who initially has a random linear combination of a random subset of $M$ messages in the database as side information. The goal of the user is to retrieve one message from the servers, while protecting the \mbox{identities} of both the demand message and the side information messages jointly. We considered two different models for this problem depending on whether the side information is a function of the demand message or not. First, we focused on the setting in which the side information is not a function of the demand message. For this setting, we proved that the capacity is given by $\left(1+{1}/{N}+\dots+{1}/{N^{K-M-1}}\right)^{-1}$. Then, we considered the setting in which the side information is a function of the demand message. For this setting, we show that the capacity is lower-bounded by $\left(1+{1}/{N}+\dots+{1}/{N^{K-M}}\right)^{-1}$. Our proposed achievability schemes are inspired by our recently proposed scheme for the single-server PIR-PCSI problem in conjunction with the scheme proposed by Sun and Jafar for multi-server private computation problem.

\bibliographystyle{IEEEtran}
\bibliography{PIRRefs}
\end{document}